\RequirePackage{fix-cm}
\documentclass[twocolumn]{svjour3}          
\smartqed  
\usepackage{graphicx}
\usepackage{amsmath,amsfonts}
\newcommand{\be}{\begin{equation}}
\newcommand{\ee}{\end{equation}}
\journalname{Journal of Superconductivity and Novel Magnetism}
\begin{document}

\title{Fundamental building blocks of strongly correlated wave functions
\thanks{Supported by the
Croatian Ministry of Science grant 119-1191458-0512 and by the University of
Zagreb grant 202759.}
}

\titlerunning{Strongly correlated wave functions}

\author{D. K. Sunko
}


\institute{D. K. Sunko \at
              Department of Physics, Faculty of Science, University of Zagreb \\
              \email{dks@phy.hr}           
}

\date{Received: date / Accepted: date}

\maketitle

\begin{abstract}

The calculation of realistic N-body wave functions for identical fermions is
still an open problem in physics, chemistry, and materials science, even for N
as small as two. A recently discovered fundamental algebraic structure of
many-body Hilbert space allows an arbitrary many-fermion wave function to be
written in terms of a finite number of antisymmetric functions called shapes.
Shapes naturally generalize the single-Slater-determinant form for the ground
state to more than one dimension. Their number is exactly $N!^{d-1}$ in $d$
dimensions. An efficient algorithm is described to generate all fermion shapes
in spaces of odd dimension, which improves on a recently published general
algorithm. The results are placed in the context of contemporary
investigations of strongly correlated electrons.

\keywords{Strong correlations \and Many-body wave functions \and Invariant
theory}

\end{abstract}

\section{Introduction}
\label{intro}

The study of strong correlations has emerged as the focal point of both
fundamental and applied research in physics, chemistry, and materials science.
The reason is that modern functional materials fall in between the standard
textbook limits of ionic and metallic (or covalent) bonding. In particular the
two currently most interesting classes of materials, the high-temperature
superconducting cuprates and pnictides, both exhibit a fascinating mixture of
ionicity and metallicity~\cite{Lazic15,Borisenko16} which remains to be
unravelled. New tools and approaches are constantly being
sought~\cite{Hyowon14} for the description of electrons which inhabit active
(open) orbitals in these materials, for which the paradigm ``strongly
correlated electrons'' has been coined long ago.

In cuprates at least, the experimental evidence points to a separation of
roles between the electrons occupying copper and oxygen orbitals, such that,
roughly speaking, the coppers are responsible for the local, and the oxygens
for the extended degrees of freedom~\cite{Niksic13}. Because of strong Cu--O
hybridization, this separation is partly a dynamical
phenomenon~\cite{OSBarisic12}, and partly produces real-space
disorder~\cite{Bianconi87a,Tahir-Kheli11,Campi14}. It leads to a picture of
network~\cite{GBianconi12} or percolation~\cite{Tahir-Kheli13} conductivity,
in which it may be possible to reconcile the local strongly correlated
behavior with Fermi-liquid transport properties~\cite{Mirzaei12}. In
particular, if the hole concentration is $1+x$, the transport properties in
the superconducting range of dopings scale with $x$, indicating that the
``$1$'' hole remains localized~\cite{Mirzaei12,Chan14a}.

Remarkably, the outlines of a similar situation can be discerned in the case
of hydrogen disulphide. It becomes superconducting at high
temperature~\cite{Drozdov15} only after undergoing a structural phase
transition~\cite{Einaga16} at $\sim 150$~GPa, which necessarily involves the
active sulphur orbitals. Similarly, a rearrangement of orbital content is
inferred for the superconducting wave function~\cite{Bussmann-Holder16}.

The principal issue in strong correlations is the need to satisfy some
dynamical restriction (e.g.\ no double occupation of a $d$ orbital)
simultaneously with the Pauli principle. The problem is that the Pauli
principle is kinematically so restrictive that little configuration space
remains for the dynamically induced correlations, so one is at a loss to
understand how the system manages to satisfy both. Indeed the weak-coupling
paradigm is so ubiquitous precisely because the system usually does \emph{not}
manage both, instead it looks almost as the non-interacting one even in the
presence of strong interactions: this is the well-known Fermi liquid.

Recently, a new description of fermion many-body states has
emerged~\cite{Sunko16-1} which promises to shed some light on the above issues
from a fundamental point of view. It turns out that every system of $N$
identical (i.e.\ spinless or spin-polarized) fermions in $d$ dimensions has a
number of special states called \emph{shapes}, which are distinguished by a
certain type of irreducibility, such that they cannot be interpreted as
consisting of lower-energy states, even when their energy is high. Although
their number is absolutely very large ($N!^{d-1}$), it is vanishingly small
compared to all possible states spanning the same energy range. The shapes
form a kind of backbone of $N$-body Hilbert space, such that every state can
be described as some superposition of \emph{bosonic} excitations of the
shapes. In other words, the shapes are the only genuinely antisymmetric
states, while all the other (infinitely many) $N$-fermion states are shapes
masked by bosons. Shapes seem to be a natural way to describe the strongly
correlated wave functions, because they are formal alternatives to the
single-Slater-determinant ground state of the weak-coupling limit. In order to
study them, one has to have a way to generate them. A new algorithm for that
purpose is described in the present article. In addition to being much more
efficient than the previously published~\cite{Sunko16-1} one, it offers some
structural insight into shapes in odd dimensions. Here it is described in
detail for the particular case of three particles in three dimensions. An
introductory review of the shape formalism can be found
elsewhere~\cite{Sunko16-2}.

\section{Efficient algorithm for fermion shapes in odd dimensions}
\label{alg}

\subsection{Previous results~\cite{Sunko16-1}}
\label{prevrev}

Consider spinless (or maximum-spin) states only. Then any antisymmetric wave
function of $N$ fermions in $d$ dimensions may be written
\be
\Psi=\sum_{i=1}^D \Phi_i(\mathbf{r}_1,\ldots,\mathbf{r}_N)
\Psi_i(\mathbf{r}_1,\ldots,\mathbf{r}_N),
\label{basic}
\ee
where $D=N!^{d-1}$, $\Psi_i$ are antisymmetric with respect to the interchange
of any two vector coordinates $\mathbf{r}_i$, while $\Phi_i$ are
symmetric in each Cartesian coordinate component of the
$\mathbf{r}_i$ \emph{separately}. The $\Psi_i$ are called \emph{shapes}.

The crucial step enabling this formulation is the classification of wave
functions by the number of single-particle nodes, which is called their grade.
Because nodes always count the degrees of freedom of the system, and the
energy is linear in the nodes for the harmonic oscillator, the sum over states
for fermions in an oscillator well becomes completely general, as soon as one
reinterprets the energy as the grade. In order to emphasize this
reinterpretation, the usual $e^{-\beta\hbar\omega}$ is denoted $q$.
Specifically, the sum over states, organized by grade, for $N$ identical
particles in $d$ dimensions reads
\begin{equation}
Z_d(N,q)=Z_E(N,q)^dP_d(N,q),\quad Z_E=\prod_{k=1}^N\frac{1}{1-q^k},
\label{zdnq}
\end{equation}
where $P_d(N,q)$ is a polynomial in $q$, called a \emph{shape polynomial},
which is the generating function of shapes by grade. It satisfies
\emph{Svrtan's recursion}
\begin{equation}
NP_d(N,q)=\sum_{k=1}^{N}(\pm 1)^{k+1}\left[C^{N}_{k}(q)\right]^dP_d(N-k,q),
\label{recurPd}
\end{equation}
with the upper sign for bosons, and the lower for fermions. Here
\be
C^N_k(q)=\frac{(1-q^N)\cdots (1-q^{N-k+1})}{(1-q^k)}
\ee
is a polynomial, and $P_d(0,q)=P_d(1,q)=1$.

One can show that the shape polynomial is symmetric in even space dimensions,
while in odd dimensions the coefficient lists in the shape polynomials for
fermions and for bosons are ``mirror images'' of each other, e.g. for $N=3$
particles in $d=3$ dimensions, they are respectively
\begin{gather}
P_3(3,q)=q^9 + 3q^7 + 7q^6 + 6q^5 + 6q^4 + 10q^3 + 3q^2,\nonumber\\
B_3(3,q)=1 + 3q^2 + 7q^3 + 6q^4 + 6q^5 + 10q^6 + 3q^7.
\label{polN3d3}
\end{gather}
This property will be called ``mirroring.''

In this approach, single-particle wave functions are represented as formal
powers, such that the exponent denotes the grade. The formal-power
representation can easily be mapped onto any concrete realization, e.g.\ for
the harmonic oscillator,
\begin{equation}
t_i^lu_j^mv_k^n\to H_l(x_i)H_m(y_j)H_n(x_k)
e^{-(x_i^2+y_j^2+z_k^2)/2}.
\end{equation}
The formal-power representation encodes the essential behavior of nodes under
multiplication and addition of functions. If two functions are multiplied, the
number of nodes is added. If the functions are added, the number of nodes is
at most the same as that of the function with the larger number of nodes. This
encoding unleashes the formidable power of classical invariant
theory~\cite{Derksen02} for the classification of many-fermion wave functions.

In Ref.~\cite{Sunko16-1} an algorithm was described to obtain all shapes for
arbitrary $N$ and $d$. Unfortunately it is quite inefficient, making it
difficult to obtain all the shapes in three dimensions already for $N=5$, even
on a very large computer. A much more efficient algorithm is described below.

\subsection{Degree of the shape polynomial}
\begin{proposition}\label{degodd}
The degree of the shape polynomial for fermions in odd dimensions and for
bosons in even dimensions is
\be
\deg P_d(N,q)=\frac{dN(N-1)}{2}\equiv D(d,N).
\label{hyp}
\ee
\end{proposition}
\paragraph{Note} The formula is also correct when $N=0$ or $1$, for which
there is no difference between fermions and bosons.
\begin{proof}
Using~\eqref{hyp} as an induction hypothesis, it follows from the
recursion~\eqref{recurPd} that
\be
\deg P_d(N,q)=\max_k\left\{D(d,N-k)+\deg C_k^N(q)^d
\right\}.
\ee
Given that
\be
\deg C_k^N(q)=\frac{k(2N-k-1)}{2},
\ee
one finds that each term in~\eqref{recurPd} has the same degree,
\be
\deg P_d(N,q)=\max_k\frac{dN(N-1)}{2}=\frac{dN(N-1)}{2},
\label{constk}
\ee
which establishes the induction step. It remains to establish the basis.
Fermions and bosons begin to differ for $N=2$, for which the recursion gives
\be
P_d(2,q)=\frac{(1+q)^d\pm(1-q)^d}{2}.
\ee
The coefficient of $q^d$ in this formula is $[1\pm(-1)^d]/2=1$ for the two
cases in the proposition, which establishes the induction basis for them,
because $D(d,2)=d$.~$\Box$
\end{proof}
\begin{proposition}\label{degeven}
Let $G(d,N)$ be the lowest nonvanishing power of the fermion shape polynomial
for given $d$ and $N$. Then the degree of the boson shape polynomial in odd
dimensions and of the fermion shape polynomial in even dimensions is
$D(d,N)-G(d,N)$.
\end{proposition}
\paragraph{Note} For fermions in an oscillator well, $G(d,N)$ is the
non-interacting ground-state energy.
\begin{proof}
By mirroring, the boson and fermion shape polynomials span the same range of
powers in odd dimensions. For the boson shape polynomial, the lowest power of
$q$ is always zero, because the boson ground-state wave function is a
constant. Hence its highest power (degree) must be shifted relatively to the
fermion polynomial by the same difference as the lowest power, which is
$G(d,N)$, so its degree is $D(d,N)-G(d,N)$. [E.g., $7=9-2$ in
Eq.~\eqref{polN3d3}.]

In even dimensions, the shape polynomial must be symmetric. By
Eq.~\eqref{constk}, each term in the recursion~\eqref{recurPd} has the same
degree $D(d,N)$, so one can say that the fermion polynomial always spans
the powers from zero to $D(d,N)$, but with some leading and trailing
coefficients equal to zero, because the corresponding powers of $q$ cancel in
the recursion. Given that it is symmetric, the number of leading and trailing
zeros must be the same, so if the first non-zero coefficient belongs to the
power $G(d,N)$, the last will belong to the power $D(d,N)-G(d,N)$.~$\Box$
\end{proof}

\subsection{Highest fermion shape in odd dimensions}\label{highest}
\begin{proposition}\label{highshape}
For fermions in odd dimensions, the highest-graded shape is unique and given
by the product of Vandermonde determinants across space dimensions.
\end{proposition}
\paragraph{Note}
This is just the product of 1D ground states for each dimension.
It is antisymmetric if and only if the number of dimensions is odd.
\begin{proof}
The bosonic ground state is nondegenerate, hence the coefficient of $q^0$ in
the boson shape polynomial is unity. In odd dimensions, the  coefficient of
$q^{D(d,N)}$ in the fermion shape polynomial is also unity by mirroring, so
the corresponding shape is unique.

The Vandermonde form is a product of linear terms $t_i-t_j$ with $1\le i<j\le
N$, so its degree is just the number of terms, $N(N-1)/2$. The total degree
of a product of $d$ such forms is $dN(N-1)/2=D(d,N)$. It is antisymmetric when
$d$ is odd, so to see that it is a shape one only needs to show that it has
no symmetric factor, i.e. cannot be written as $\Phi\Psi$ with some symmetric
$\Phi\neq 1$. This is obvious, because it is a product of linear antisymmetric
terms only. Because the shape of degree $D(d,N)$ is unique, the stated product
of Vandermonde determinants is that shape.~$\Box$
\end{proof}

\subsection{Lowering the grade of a shape}

It is easy to lower the number of nodes of any wave function in the abstract
formal-power representation. One simply lowers the degree of the polynomial
representing it. The \emph{shift operators} serve this purpose:
\begin{gather}
T_k^m(\cdots t_k^n\cdots)=(\cdots t_k^{n+m}\cdots),\nonumber\\
\bar{T}_k^m(\cdots t_k^n\cdots)=\left\{
\begin{tabular}{cl}
$(\cdots t_k^{n-m}\cdots)$ & $n\ge m,$\\
$0$&$n<m$.
\end{tabular}
\right.
\end{gather}
Here parentheses denote any monomial. The shift operator corresponding to any
given variable ($t_k$ above) is denoted by capitalizing the same letter.
Shifting ``down'' is denoted by the overbar. Shift operators are linear, i.e.\
they distribute naturally over polynomials. Like derivative operators, the
downshifts do not commute with the upshifts. For example, $T\bar{T}1=0$
but $\bar{T}T1=1$.
\begin{proposition}\label{shiftdet}
A shift operator acting on any determinant in which its corresponding variable
appears in a single column acts by shifting all powers of that variable in
that column simultaneously.
\end{proposition}
\begin{proof}
Expand by that column.~$\Box$
\end{proof}
For example,
\be
\bar{V}_1\left|\begin{matrix}
t_{1}^2u_{1}^3v_{1}&
t_{2}^2u_{2}^3v_{2}\\
t_{1}u_{1}^2&
t_{2}u_{2}^2
\end{matrix}
\right|=
\left|\begin{matrix}
t_{1}^2u_{1}^3&
t_{2}^2u_{2}^3v_{2}\\
0&
t_{2}u_{2}^2
\end{matrix}
\right|=t_{1}^2u_{1}^3t_{2}u_{2}^2.
\ee
Clearly the action of a shift on a Slater determinant does not give a Slater
determinant. The power of Proposition~\ref{shiftdet} is that one can iterate
the prescription, i.e.\ apply it to the resulting determinant, nevertheless.
The idea is to use shifts to make lower-grade shapes from the highest one.
Because a simple shift does not preserve antisymmetry, we shall use
symmetrized shifts, denoted by an underline:
\be
\underline{A_iB_jC_k\cdots}=\sum_{m=0}^{N-1}
A_{i+m}B_{j+m}C_{k+m}\cdots,
\ee
where the particle indices on the right are understood modulo $N$. The index
$1$ in symmetrized shifts is understood, e.g. we write
$\underline{\bar{T}_1\bar{U}_1}$ as $\underline{\bar{T}\bar{U}}$. Note that
$\underline{A_iB_j}\neq\underline{A_{i\vphantom{j}}}\,\underline{B_j}$.

Specialize to $d=3$ now, with formal variables $t_i,u_i,v_i$, $i=1,\ldots,N$.
Then the highest-graded shape is $S\equiv\Delta_N(t)\Delta_N(u)\Delta_N(v)$, in
obvious notation, to be called the \emph{source shape} in the following.
\begin{proposition}\label{zeroshift}
Let $\Delta_N(t)$ be the Vandermonde form in the variables $t_1,\ldots,t_N$.
Then
\be
\underline{\bar{T}}\Delta_N(t)=0.
\ee
\end{proposition}
\begin{proof}
This is a cyclic sum of alternating terms.~$\Box$
\end{proof}
Proposition~\ref{zeroshift} is good news, because one can show~\cite{Sunko16-1}
that there are no fermion shapes of next-to-highest grade in odd dimensions,
cf.\ Eq.~\eqref{polN3d3}. In other words, it appears that the downshifts cannot
leave the space of shapes, if applied iteratively to the source shape. This
idea is at the core of the efficient algorithm to generate shapes.

\subsection{Description of the algorithm}

\begin{figure*}
\begin{center}
\includegraphics[width=0.75\textwidth]{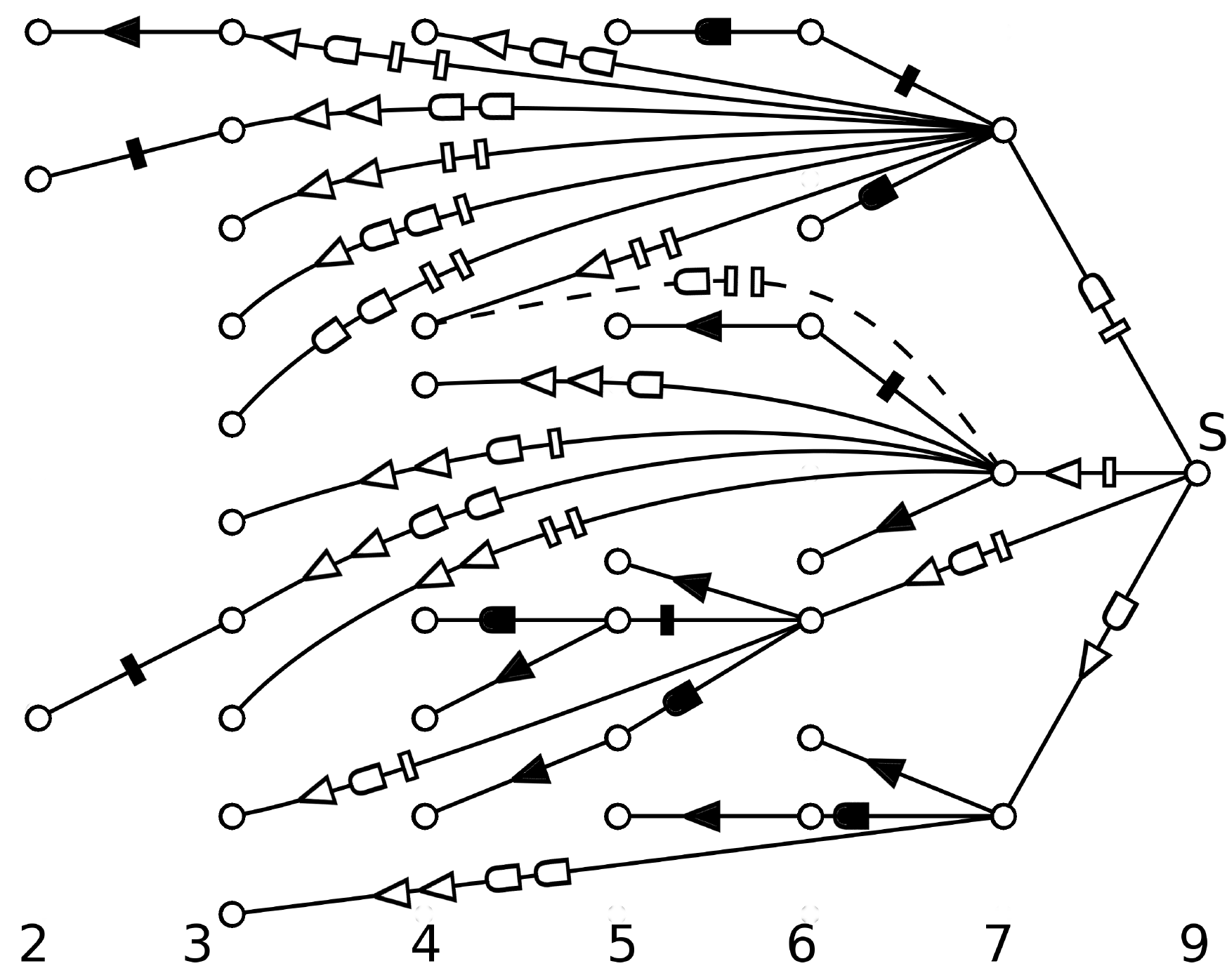}
\end{center}
\caption{Branching tree of shapes for $3$ fermions in $3$ dimensions. The $36$
shapes are nodes, arranged in columns, with grade given below each column. The
edges are decorated with operators whose action traverses the edge from right
to left. The operators $\bar{T}$, $\bar{U}$ and $\bar{V}$ are depicted by
graphically similar open symbols. The closed symbols analogously represent the
operators $\underline{T\bar{T}}^2$, $\underline{U\bar{U}}^2$ and
$\underline{V\bar{V}}^2$. All operators on a given edge are symmetrized
together. For example, the shape at upper left is given by
$\underline{V\bar{V}}^2\,\underline{\bar{V}\bar{U}\bar{T}}^2\,
\underline{\bar{U}\bar{T}}\,S$, where $S$ is the source shape, depicted at far
right. The dashed line is explained in the text.}
\label{graph}
\end{figure*}
The algorithm will now be described on the specific example of $3!^2=36$
shapes of $3$ particles in $3$ dimensions. In Fig.~\ref{graph}, the shapes are
depicted as nodes in a graph, with the source shape
$S=\Delta_3(t)\Delta_3(u)\Delta_3(v)$ on the right. Transformations of one
shape into another are depicted by edges of the graph. These are effected by
lowering operators, as denoted by edge decorations in the figure, giving the
edges a natural orientation from right to left, also depicted by the
orientation of the arrow-like symbols. Notably, a simple lowering operator
like $\underline{\bar{T}}$ \emph{always} gives zero when acting on any shape
(cf.\ Proposition \ref{zeroshift}), so the elementary operators which lower
the grade by one are $\underline{T\bar{T}}^2$ and similar, depicted by filled
symbols. All shapes are generated from the source by lowering operators. One
can imagine the operator symbols on the edges as filters, or funnels, which
take the source flow from right to left, letting through shapes of ever lower
grade.

\subsection{Fermion sign problem}

The three directions in space are equivalent, and so are the shift operators
corresponding to them. Permutations of the shift operators in Fig.~\ref{graph}
give rise to different but equivalent graphs. The edges depicted by full
lines form a particular kind of oriented spanning tree, where every node
except the source (root of the tree) has exactly one incoming edge, while the
root has none. Every choice of such a \emph{branching tree} obviously fixes
the phases of all shapes uniquely. It may be possible to add edges
consistently with this sign choice, but this cannot be guaranteed in general.
A conflicting insertion is depicted by the dashed line: in equations, it turns
out that
\be
\underline{\bar{V}\bar{T}}^2\,\underline{\bar{U}\bar{T}}\,S=
-\underline{\bar{U}\bar{T}}^2\,\underline{\bar{V}\bar{T}}\,S.
\label{sign}
\ee
This observation means that one cannot simply ``turn loose'' all possible
operators on the source state to generate all possible shapes, because one
will encounter the fermion sign problem\cite{Loh90}. In other words, a
context-free, or local, definition of shape signs is not possible, because any
algorithm changing the states will in principle allow some local moves which
spoil the agreed-upon signs. Instead, the correct algorithmic definition of
shape signs is a choice of branching tree rooted at the source, which is a
\emph{global} object.

In principle, simulations can deal with the above sign issue in one of three
ways. The first is to generate all shapes beforehand and use them as a basis,
while varying only the coefficients $\Phi_i$ in the simulation. This approach
naturally leads to representing physical states as ``vectors of symmetric
polynomials,''
\be
\left(\Phi_1,\Phi_2,\ldots,\Phi_{N!^{d-1}}\right).
\label{fmod}
\ee
Such a structure is called a \emph{free module} (as distinct from a vector
space, where the $\Phi_i$ would be just numbers). It evidently solves the
sign problem, because states are mapped to a space of symmetric functions. The
practicality of this proposal remains to be demonstrated.

Another possible approach is to compile a list of allowed operators, which are
consistent with a given branching tree. These operators could then be used in
a context-free manner, enabling one to generate shapes ``on the fly'' without
storing them explicitly. It is an open question at present whether such a set
of mutually consistent operators can always be found, which is also complete
in the sense that they generate all the shapes.

Finally, one can try to find rules of calculation with the operators involved.
In this approach, the individual shift operators are letters, while the
symmetrized operators --- underlined strings of one or more letters --- are
words. The task is to find the grammar of this language, a sort of extended
Wick's theorem. From this point of view, Eq.~\eqref{sign} looks as if the
letters $\bar{T}$, $\bar{U}$ and $\bar{V}$ were anticommuting. The previous
question of finding a complete consistent set of operators may now be
rephrased: can one compile a list of words such that using them does not
require a grammar? The formal-language approach is potentially the most
powerful way to manipulate many-body states, but also requires the most future
research.

\section{Discussion}

In the present work an efficient algorithm has been described, which generates
all shapes of $N$ particles in odd dimensions. Much about the algorithm and
especially the branching-tree structures it naturally engenders remains to be
clarified. The discussion here places it in the broader context of efforts to
represent fermion systems efficiently, concentrating on the open questions.

Most pragmatically, one can regard the algorithm as just another way to obtain
shapes, more practical than the other known~\cite{Sunko16-1} one, but in any
case a means to an end. With the shapes in hand, the really interesting
insight is to represent physical states as a free module~\eqref{fmod}, rather
than a vector space. This is in some sense the furthest one can take
Heisenberg's matrix mechanics. It explains immediately why fermion systems
cannot be directly bosonised in more than one dimension~\cite{Tomonaga50}.
Namely, in one dimension there is only one shape, the ground-state Slater
determinant $\Psi_0$, so that any state can be written as $\Psi=\Phi\Psi_0$.
Because $\Phi$ is a symmetric function, bosonisation succeeds: every excited
state $\Psi$ is uniquely mapped on some boson wave function $\Phi$. In the
standard second-quantized formalism, this result reads, say,
\be
\left|\Psi\right>=B_1^\dagger B_2^\dagger\left|\Psi_0\right>
\leftrightarrow\Phi_1\Phi_2\Psi_0=\Phi\Psi_0=\Psi,
\ee
for a given product of boson excitations. Because the free module~\eqref{fmod}
is one-dimensional in one dimension, the structure of excitations is purely
multiplicative. Generally, however, the free module has dimension $N!^{d-1}$,
so that excitations can be, for example,
\be
B_1^\dagger\left|\Psi_1\right>+
B_2^\dagger\left|\Psi_2\right>
\leftrightarrow\Phi_1\Psi_1+\Phi_2\Psi_2,
\ee
with \emph{same} $B_i$'s (symmetric polynomials $\Phi_i$, or bosons) but
\emph{different} $\Psi_i$'s (shapes, or vacua). One can say either that
bosonisation fails, because the structure of excitations is no longer
multiplicative, or that it finally succeeds, because one has found the correct
generalization of the one-dimensional case. In any case, the ``deep''
structure of fermionic excitations exposed here is that the vacua are like
prime numbers, in the sense that they do not factorize: one cannot be obtained
from another by multiplication. Therefore excitations must be described by a
combination of multiplication \emph{and addition}. As of this writing, it is
of greatest interest to learn to calculate efficiently in the free module,
because mapping fermionic states onto symmetric functions \emph{a priori}
solves the fermion sign problem.

The algorithm has an interesting feature from the theoretical point of view.
All its moves \emph{reduce} information, because they are net downshifts,
which correspond to lowering monomial powers, reducing the overall degree of
the polynomials involved. In order to go in the opposite direction, raising
the degree, one would have to use quite ``clever'' combinations of upshifts in
order to stay within the space of shapes, i.e.\ avoid states of the general
form~\eqref{basic} with some $\Phi_i\neq 1$. In other words, upshifting
requires \emph{adding} information in order to make higher-grade shapes from
lower-grade ones. It is like integration, while downshifting is like taking
derivatives: one requires insight, while the other is an automatic operation.
One must conclude that the source shape has the maximum information content,
so that the ``flow'' passing through ``filters'' in Fig.~\ref{graph} is the
flow of information, or negentropy.

This conclusion runs quite counter to thermodynamic intuition, which takes for
granted that states with high excitation energy have high entropy as well. The
critical issue in this reasoning is the relationship between the number of
nodes and the energy of the state. If the state is dominated by kinetic
energy, one is in the weak-coupling limit, and the usual thermodynamic
reasoning prevails. However, if it is dominated by correlations, the system
may choose a ``complicated'' ground state, with more nodes, but unique in some
sense, hence of low entropy. This situation is called strongly correlated, the
most famous example being Hund's rule~\cite{Yamanaka05}.

The shape paradigm provides an interesting way to think about the strongly
correlated limit. It is as if the system stays cold by using extra nodes
to store information, in the form of some rare complicated states, instead of
assigning nodes to kinetic motion, which would distribute them among a large
number of common simple states, with high entropy. In particular, the source
state is unique among a very large number of states with the same number of
nodes. In our example of three particles in three dimensions, there are $3838$
states with nine nodes, only one of which is the source. In fact, mirroring
indicates there must be a way to think of the source as a zero-entropy state,
equivalent to the completely featureless boson ground state. Its concrete
realization as a product of three one-dimensional fermionic ground states
indicates the same.

A simple way to reconcile the above discussion with standard thermodynamics is
to assign to each shape an entropy given by the logarithm of the coefficient
of the shape polynomial, corresponding to its grade. This resolution has the
pleasant property of specializing to the usual definition of entropy of the
non-interacting ground state, which is just the logarithm of its degeneracy
($\ln 3$ in Fig.~\ref{graph}). The source shape always has zero entropy, just
as the reasoning above indicated it should. In this way one can think of
shapes as low-entropy states embedded in a much larger space of high-entropy
ones. The latter are described by bosonic excitations of the shapes, as given
by Eq.~\eqref{basic} with some $\Phi_i\neq 1$. In other words, the proper
physical resolution of the above conundrum is that the shapes are a choice of
possible vacua for a physical system, and these vacua are special in the sense
that they have an exceptionally low entropy, or degeneracy, for their given
energy. Once the ground state is selected, perhaps as a superposition of the
vacua, the remaining shapes may still make their presence felt as bandheads of
higher-energy excitation bands, such as are ubiquitous in the spectra of
finite systems. In this way their ``exceptionalism'' persists, giving them a
special role in the excitation spectrum, even if some other state is the
ground state~\cite{Sunko16-2}.

\section{Conclusion}

The shape paradigm has promise as both a theoretical and practical tool for
the description of strongly correlated finite systems, particularly of
fermions. While much remains to be done, the algorithm described in the
present work removes a major roadblock in the practical application of the
paradigm to transition-metal compounds, whose open $3d$ orbital requires that
one should be able to manipulate states of around $N=5$ identical fermions.
These materials are in the focus of current fundamental and applied interest,
as both cuprate and pnictide high-temperature superconductors belong to this
category. It is possible to separate the local (strongly correlated) part of
the problem from the extended one~\cite{Hyowon14}, making shapes an
interesting contender for the description of the former. It is still too early
for a direct comparison of the shape paradigm with other more mature
approaches, or with experiment. Hopefully the readers will be motivated to
join the exploration of shapes based on their own interests.

\begin{acknowledgements}
I thank D.~Svrtan for his help and interest.
\end{acknowledgements}


\end{document}